\documentclass[submission,copyright,creativecommons]{eptcs}
\providecommand{\event}{PxTP 2017} 
\usepackage{underscore}           

\usepackage{amsthm}

\usepackage{amsmath}
\usepackage{amstext}
\usepackage{amssymb}
\usepackage{amsfonts}
\usepackage{xspace}
\usepackage{graphicx}
\usepackage{url}
\usepackage{subfigure}
\usepackage{mdwlist}
\usepackage{courier}
\usepackage{lscape}
\usepackage{comment}
\usepackage{wrapfig}
\usepackage{multirow}
\usepackage{bussproofs}
\usepackage{pdftricks}
\usepackage{subfigure}
\begin{psinputs}
\usepackage{pstricks}
\usepackage{color}
\usepackage{pstcol}
\usepackage{pst-plot}
\usepackage{pst-tree}
\usepackage{pst-eps}
\usepackage{multido}
\usepackage{pst-node}
\usepackage{pst-eps}
\end{psinputs}
\usepackage{cite}
\makeatletter
\newif\if@restonecol
\makeatother

\usepackage{footmisc}
\usepackage{balance}
\usepackage{enumerate}
\usepackage[normalem]{ulem}
\usepackage{hyperref}
\usepackage{alltt}

\usepackage{epsfig}
\usepackage{manfnt}
\usepackage{graphicx}
\usepackage{color}
\usepackage{multirow}
\usepackage{tabularx,colortbl}
\usepackage{alltt}
\usepackage{bussproofs}
\usepackage{algorithm}
\usepackage{algorithmicx}
\usepackage{pifont}
\usepackage{cancel}
\usepackage{rotating}

\usepackage{float,lipsum}
\floatstyle{boxed}
\restylefloat{table}
\restylefloat{figure}

\newcommand{\formulae}{formul\ae\xspace}

\newcommand{\dpll}[1]{{\sc DPLL}\xspace}

\newtheorem{definition}{Definition}
\newtheorem{proposition}{Proposition}
\newtheorem{theorem}{Theorem}

\newcommand{\COMMENT}[1]{}

\newcommand{\ux}{\ensuremath{\underline x}}
\newcommand{\uu}{\ensuremath{\underline u}}

\newcommand{\uz}{\ensuremath{\underline z}}

\newcommand{\cM}{\ensuremath \mathcal M}

\newcommand{\cT}{\ensuremath \mathcal T}
\newcommand{\cS}{\ensuremath \mathcal S}
\newcommand{\cC}{\ensuremath \mathcal C}

\newcommand{\cI}{\ensuremath \mathcal I}

\newcommand{\val}[1]{\ensuremath #1_{\cM, \cI}}
\newcommand{\valp}[1]{\ensuremath #1_{\cM, \cI'}}

\newcommand{\Proc}{\ensuremath \mathtt{Proc}}
\newcommand{\Data}{\ensuremath \mathtt{Data}}

\renewcommand{\int}{\ensuremath {\mathcal I}}

\newcommand{\tv}{\ensuremath{{\bf v}}\xspace}

\title{Counter Simulations via Higher Order \\ Quantifier Elimination: a preliminary report }
\author{Silvio Ghilardi
 \institute{ Universit\`a degli Studi di Milano, Milano, Italy\thanks{The first authar was supported by the INdAM's GNSAGA group.}}
\and
Elena Pagani
 \institute{ Universit\`a degli Studi di Milano, Milano, Italy}
}

\def\titlerunning{Counter Simulations}
\def\authorrunning{S. Ghilardi, E. Pagani}

\begin{document}
\maketitle

\begin{abstract}
\setcounter{footnote}{0}
 Quite often, verification tasks for distributed systems are accomplished via counter  abstractions. Such abstractions can sometimes be justified via 
simulations and bisimulations. In this work, we supply  logical foundations to this practice, by  a specifically designed technique for second order quantifier elimination. Our method, once applied to  specifications of verification problems for 
parameterized distributed 
systems, produces integer variables systems that are ready to be model-checked by current SMT-based tools. We demonstrate the feasibility of the approach with a prototype implementation and first experiments. 
\end{abstract}

\section{Introduction}~\label{sec:introduction} 

In this paper we introduce a methodology moving
from higher order specifications down to 
simulations expressible inside first-order theories, where SMT techniques can be effectively applied. We believe that this methodology, requiring user intervention only for initial choices at design phase, can supply a good example of the interaction between  logic
engines operating at different expressivity levels. The motivation of our research lies in the area of the verification of distributed (especially fault-tolerant) algorithms, where benchmarks for our first experiments were taken from.

The automated, formal verification of distributed algorithms is a crucial, although challenging, task.  
The processes executing these algorithms communicate with one another, their actions depend on the messages received, and their number is arbitrary.  These characteristics are captured by so called reactive parameterized systems.  The task of validating or refuting properties of these systems is daunting, due to the difficulty of limiting the possible evolutions, thus having to deal with genuinely infinite-state systems.   

Building accurate \emph{declarative}  models of these systems requires powerful formalisms, involving arrays~\cite{GhilardiR10a},\cite{GhilardiR10} and, in the fault-tolerant case, also some fragment of higher order logic~\cite{dragoj},\cite{arca} (this is needed in order to have some form of comprehension to handle cardinalities of  definable sets). On the other hand, for a long time, it has been observed that 
 {\em counter 
systems
}~\cite{delzannoFMSD, bro1,bro2} can be sufficient to specify many problems (like cache coherence or broadcast protocols) in the distributed algorithms area. Recently, 
counter abstractions have been effectively used also in the verification of fault-tolerant distributed protocols~\cite{konnov2013A,konnov2013B,KonnovVW15,cilc16}.
It should be noticed that, unlike what happens in the old framework of~\cite{delzannoFMSD, bro1,bro2}, these new applications are often (although not always) based on abstractions that can only \emph{simulate} the original algorithms and such simulation may sometimes be the 
result of an a priori reasoning on the characteristics of the algorithm, embedded into the model. 
Despite this fact, all  runs from the original specifications are represented in the simulations with counter systems
(this is in fact the formal content of the notion  of a `simulation'), thus for instance safety certifications for the simulating model apply also to the original model. The advantage of this approach is that, as it is evident e.g. from the experiments in~\cite{cilc16}, \emph{verification of counter systems 
is very well
supported by the existing technology}. In fact, although basic problems about counter systems are themselves undecidable,
the sophisticated machinery (predicate abstraction~\cite{FlanaganQ02}, IC3~\cite{ic3-tn,pdr}, etc.) developed inside the SMT community lead to impressively performing tools like $\mu Z$~\cite{muZ}, nuXmv~\cite{nuXmv}, SeaHorn~\cite{seahorn}, \dots{} which are nowadays being used to 
solve  many verification problems regarding counter systems.

Being conscious that building such simulations requires in any case some human interaction, we tried to build in this paper a uniform framework. Our framework relies on recent powerful techniques for deciding cardinality and array constraints~\cite{KunkakJAR,arca,JAR15}; as pointed out in~\cite{KunkakJAR,fmsd17,prep}, sometimes such decision techniques can   be modified so as to supply quantifier elimination results and, via
 these quantifier elimination results, we shall show how 
to \emph{automatically build the best possible counter simulations} users can obtain  once they fixed (i) the specification of the system, (ii) possibly some helpful invariants and (iii) the counter variables involved in the projected  simulation (such variables  are cardinality counters for definable sets). We demonstrate  
the effectiveness of our approach by producing, for some  benchmarks, counter systems simulations which are effectively model-checked by current SMT-based tools. 

\subsection{A four-steps strategy}\label{subsec:strategy}

\emph{Our general four-steps strategy  can be summarized as follows:}
 (1) system specifications (together with their safety problems) are formulated in higher order logic, i.e. using a declarative formalism which is sufficiently expressive and close to informal specifications; (2) counters for definable sets are added by the user to the system specification, in such a way that the observationally relevant properties can be reformulated as arithmetic properties of these counters; (3) higher order variables are eliminated, by applying an automatic procedure; (4) the resulting system is finally model-checked by using an SMT-based tool for counter systems. The reader is referred to Section~\ref{app:onethird} for a detailed example.
 
 In this plan, only steps (1) and (2) require manual 
intervention; step (3) is effective every time the syntactic restrictions for  quantifier elimination procedures are matched; step (4) is subject to two risks, namely to the fact that model-checkers may not terminate on such (undecidable) problems and to the fact that simulations may introduce spurious traces. 
Non-termination, giving the actual state of the art (much progress has been made both at the theoretical and at the practical level) is less frequent than one can imagine and   there are also positive theoretical
results - both classical~\cite{lics,GhilardiR10a} and more recent~\cite{KonnovVW17} - that guarantee termination in some interesting cases.
Concerning the second risk, notice 
that if spurious traces arise, they can be recognized because SMT tools supply concrete values for counterexamples; then,
 one can go back to step (2) and  refine the abstraction by adding more counters.
 
The paper is structured as follows: Section~\ref{sec:specifications} gives general foundations; Section~\ref{sec:simulations} outlines the  formalizations  we use and supplies 
a basic quantifier elimination result; Section~\ref{app:onethird} analyzes a concrete benchmark; Section~\ref{sec:implementation} describes our implementation and our first experiments. Section~\ref{sec:conclusions} concludes.

\section{System Specifications in Higher Order Logic}\label{sec:specifications}

The behavior of a computer system can be modeled through a {\em transition system}, which is a tuple
$$\mathcal{T} = 
(W, W_0, R, AP, V)$$
such that (i) $W$ is the set of possible configurations, (ii) $W_0 \subseteq W$ is the set of initial configurations, (iii) $AP$ is a set of `atomic propositions', 
(iv) $V: W\longrightarrow AP$ is a function labeling each state with the set of propositions `true in it', 
(v) $R \subseteq W \times W$ is the transition relation:
$w_1 R w_2$ 
describes how the system can `evolve in one step'. 

\begin{definition}\label{def:simulation}
 We say that the transition system $\mathcal{T}'=(W', W'_0, R', AP, V')$ 
\emph{simulates} the transition system $\mathcal{T}=(W, W_0, R, AP, V)$ (notice that $AP$ is the same in the two systems) iff there is a relation $\rho\subseteq W\times W'$ 
(called \emph{simulation}) such that 
\begin{description}
 \item[{\rm (i)}] for all $w\in W$ there is $w'\in W'$ such that $w\rho w'$; 
 \item[{\rm (ii)}] if $w\rho w'$ and $w\in W_0$, then $w'\in W'_0$; 
 \item[{\rm (iii)}] if $w\rho w'$ and
$wR v$, then there is $v'\in W'$ such that  $w'R' v'$ and $v\rho v'$;
\item[{\rm (iv)}] if $w\rho w'$, then $V(w)= V'(w')$; 
\end{description}
 If the converse $\rho^{op}$  of $\rho$ is also a simulation, then $\rho$ is said to be a \emph{bisimulation} and $\mathcal{T}'$ and $\mathcal{T}$ are said to be \emph{bisimilar}. 
\end{definition}

Bisimilar systems are equivalent in the sense that the properties expressible in common 
temporal logic specifications (e.g. in $CTL, LTL, CTL^*$, etc.) are invariant under bisimulations; simulation is also useful as important properties (like safety properties, or more generally properties expressible in sublogics like $ACTL$) can be transferred from a system to the systems simulated by it (but not vice versa). 

We write $\cT \leq \cT'$ iff $W\subseteq W'$ and the inclusion is a simulation. This relation is a partial order;  notice that if $\cT'$ simulates $\cT$ and $\cT'\leq \cT''$, then $\cT''$ also simulates $T$; in this case, the simulation supplied by $\cT'$ is said to be \emph{stronger} or \emph{better} than the simulation supplied by $\cT''$ (in fact, one has more chances of establishing e.g. a safety property of $\cT$ by using $\cT'$ than by using $\cT''$). 

The above formalism of transition systems is often too poor, because it cannot cover rich features arising in concrete applications. To have enough expressive power, we use higher order logic, more specifically \emph{Church's type theory} (see e.g.~\cite{Andrews} for an introduction to the subject).\footnote{
Some notation we use might look slightly non-standard; it is similar to the notation of~\cite{LambekScott}. 
}
It should be noticed, however, that our primary aim is \emph{to supply a framework for model-checking and not to build a deductive system}.
Thus we shall introduce below only  suitable languages (via  higher order signatures) and a semantics for such languages - such semantics can be specified e.g. inside
any 
classical
foundational system for set theory.
In addition, as typical for model-checking, we want 
to constrain our semantics so that certain sorts have a fixed meaning: the primitive sort $\mathbb Z$ has to be interpreted as the (standard) set of  integers, the sort $\Omega$ has to be interpreted as the set of truth values $\{ \mathtt{tt}, \mathtt{ff}\}$; moreover, some primitive sorted operations like $+, 0, S$ (addition, zero, successor for natural numbers) and $\wedge, \vee, \to, \neg$ (Boolean operations for truth values) must have their natural interpretation. 
Some sorts
might be \emph{enumerated}, i.e. they must be interpreted as a specific finite `set of values'  $\{ \mathtt{a_0},\dots, \mathtt{a_k}\}$, where the 
$\mathtt{a_i}$ are mentioned among the constants of the language and are assumed to be distinct. Finally, we may ask for
a primitive sort to be interpreted as a \emph{finite set} (by abuse, we shall call such sorts \emph{finite}): for instance, we shall constrain in this way the sort $\Proc$ modeling the set of processes in a distributed system.
In  addition, if a sort is interpreted into a finite set, we may constrain some numerical parameter (usually, the parameter we choose for this is named  $\mathtt N$) to indicate the cardinality of such finite set. The notion of constrained signature below incorporates all the above requirements in a general framework.

A \emph{constrained signature} $\Sigma$ consists of a set of (primitive) sorts and of a  set of (primitive) sorted function symbols,\footnote{These include  0-ary function symbols, called constants; constants of sort $\mathbb{Z}$ will be called (arithmetic) \emph{parameters}.}
 together with a class $\cC_{\Sigma}$ of $\Sigma$-structures, called the \emph{models} of $\Sigma$.\footnote{
 In the standard model-checking literature $\cC_{\Sigma}$ is a singleton; here we must allow \emph{many} structures in  $\cC_{\Sigma}$, because our model-checking problems are \emph{parametric}: the sort modeling the set of processes of our system specifications must be interpreted onto a finite set whose cardinality is not a priori fixed. Our definition of a `constrained signature' is analogous to the definition of a `theory' in SMT literature; in fact, in SMT literature, a `theory' is just a pair given by a signature and a class of  structures.
 When transferred to a higher order context, such definition coincides with that of a  `constrained signature' above (thus our formal framework is very similar to e.g. that of~\cite{reynolds15}).
 } 
 Using primitive sorts,
 \emph{types} can be built up using exponentiation (= functions type);
 \emph{terms} can be built up using variables, function symbols, as well as $\lambda$-abstraction and functional application. 

 Our constrained signatures always include the sort $\Omega$ of truth-values;
terms of type $\Omega$ are called \emph{formulae} (we use greek letters $\alpha, \beta, \dots,\phi, \psi, \dots$ for them). For a type $S$, the type $S\to\Omega$ is indicated as $\wp(S)$ and called the \emph{power set}  of $S$; if $S$ is constrained to be interpreted as a finite set,  $\Sigma$ might contain a cardinality operator $\sharp: \wp(S) \longrightarrow  \mathbb{Z}$, whose interpretation is assumed to be the intended one ($\sharp s$ is the number of the elements of $s$ - as such it is always a nonnegative number). If $\phi$ is a formula and $S$ a type, we use $\{ x^S \mid \phi\}$ or just $\{ x \mid \phi\}$
for $\lambda x^S \phi$. We assume to have binary equality predicates for each type; universal and existential quantifiers for \formulae can be introduced by standard abbreviations (see e.g.~\cite{LambekScott}).
We shall use the roman letters $x, y, \dots, i,j,\dots, v, w, \dots$ for variables
(of course, each variable is suitably typed, but types are left implicit if confusion does not arise). Bold letters like $\tv$ (or underlined letters like $\ux$) are used for tuples 
of free variables;
below, we indicate with $t(\tv)$ the fact that the term $t$  has free  variables included in the list $\tv$
(whenever this happens, we say that $t$ is a $\tv$-term, or a $\tv$-formula if it has type $\Omega$).
%
The result of a simultaneous substitution 
of the tuple of variables $\tv$ by the tuple of (type matching) terms $\uu$ in $t$ is denoted by $t(\uu/\tv)$ or directly as $t(\uu)$.

 Given a tuple of  variables $\tv$, 
a \emph{$\Sigma$-inter\-pre\-tation} of $\tv$ in a model $\cM\in \cC_{\Sigma}$ is a function $\cI$ mapping each variable onto an element of the correponding
type (as interpreted in $\cM$).
The evaluation of a term $t(\tv)$ according to $\cI$ is recursively defined in the standard way and is written as $\val{t}$. A $\Sigma$-formula
$\phi(\tv)$ is \emph{true} under $\cM, \cI$ iff it evaluates to $\mathtt{tt}$ (in this case, we may also say that $\val{\tv}$ \emph{satisfies} $\phi$); $\phi$ is
\emph{valid} iff it is true  for all models $\cM\in \cC_{\Sigma}$ and all interpretations $\mathcal I$ of  $\tv$ over $\cM$. 
We write $\models_{\Sigma} \phi$ (or just $\models \phi$) 
to mean that $\phi$ is valid and $\phi\models_{\Sigma} \psi$ (or just $\phi\models \psi$) to mean that $\phi\to \psi$ is valid;  we say that $\phi$ and $\psi$ are $\Sigma$-\emph{equivalent} (or just equivalent) iff $\phi\leftrightarrow \psi$ is valid.

Constrained signatures are used for our system specifications as follows:

\begin{definition}\label{def:ss}
 A \emph{system specification} $\cS$ is a tuple
 $$
 \cS~=~(\Sigma, \tv, \Phi, \iota, \tau, AP)
 $$
 where (i) $\Sigma$ is a constrained signature, (ii) $\tv$ is a tuple of variables, (iii) $\Phi,\iota$ are $\tv$-\formulae and $AP$ is a set of $\tv$-\formulae, (iv) $\tau$ is a $(\tv, \tv')$-formula (here the $\tv'$ are renamed copies of the $\tv$) such that 
 \begin{equation}\label{eq:invariant}
  \iota(\tv)\models_{\Sigma} \Phi(\tv), \qquad \Phi(\tv)\wedge \tau(\tv, \tv')\models_{\Sigma} \Phi(\tv')~~~. 
 \end{equation}
\end{definition}

In the above definition, the $\tv$ are meant to be the variables specifying the system status, $\iota$ is meant to describe initial states, 
$\tau$ is meant to describe the transition relation and the  $AP$ are the `observable propositions' we are interested in. The $\tv$-formula 
$\Phi$, as it is evident from~\eqref{eq:invariant}, describes an invariant of the system (known to the user).
Of course, using the expressive power of our type theory, it would be easy to write down  the `best possible' invariant describing in a precise way the set of reachable states; however, the $\tv$-formula for such invariant might involve logical constructors (like fixpoints) lying outside the tractable fragments we
plan to use. On the other hand, invariants are quite useful - and often essential - in concrete verification tasks, this is why we included them in Definition~\ref{def:ss}.

It is now clear how to associate a transition system with any system specification:

\begin{definition}
 The transition system  of the system specification $
 \cS~=~(\Sigma, \tv, \Phi, \iota, \tau, AP)
 $ is the transition system $T^\cS$ given by $(W^\cS, W_0^\cS, R^\cS, AP^\cS, V^\cS)$, where: (i) the set of states $W^\cS$ is the set of the tuples $\val{\tv}$ satisfying $\Phi(\tv)$, varying $\cM, \cI$ among the
 $\Sigma$-models and $\Sigma$-interpretations of $\tv$; (ii) $W^\cS_0$ is the set of states satisfying $\iota(\tv)$;
 (iii) $R^\cS$ contains the couples of states $\val{\tv},\valp{\tv'}$\footnote{Notice that $\cM$ is the same;
 $W^\cS$ might be a proper class, but to avoid this it is sufficient to ask for the set of models $\cC_{\Sigma}$ of our constrained signature
 $\Sigma$ to be a \emph{set} (not a proper class).
 } satisfying $\tau(\tv, \tv')$; (iv) $AP^\cS$ is $AP$; (v) for $\psi(\tv)\in AP^\cS$, we have that $V(\psi)$ contains precisely the states satisfying $\psi(\tv)$.
\end{definition}

\section{Simulations and Counter Abstractions}\label{sec:simulations}

Model-checking a transition system like $T^\cS$ might be too difficult, this is why it could be useful to replace it with a (bi)similar, simpler system: in our applications, we shall try to replace $\cS$ by some ${\cS'}$ whose variables are all integer variables. To this aim, we `project' $\cS$ onto a subsystem $\cS'$, i.e. onto a system comprising only some of the variables of $\cS$. 

In order to give a precise definition of what we have in mind, we must first consider subsignatures: here  a \emph{subsignature} $\Sigma_0$ of $\Sigma$ is a signature obtained from $\Sigma$ by dropping some symbols of $\Sigma$ and taking as $\Sigma_0$-models the class $\cC_{\Sigma_0}$ of the restrictions $\cM_{\vert \Sigma_0}$ to the 
$\Sigma_0$-symbols of the structures $\cM\in \cC_{\Sigma}$.

\begin{definition}\label{def:sub}
 Let  $
 \cS=(\Sigma, \tv, \Phi, \iota, \tau, AP)
 $ be a system specification; a \emph{sub-system specification} of it is a system specification 
 $
 \cS_0= (\Sigma_0, \tv_0, \Phi_0, \iota_0, \tau_0, AP_0)
 $
 where $\Sigma_0$ is a subsignature of $\Sigma$, $\tv_0\subseteq \tv$, $AP_0=AP$ and we have 
 \begin{equation}\label{eq:sub}
  \Phi(\tv)\models_{\Sigma} \Phi_0(\tv_0), ~~~~\iota(\tv)\models \iota_0(\tv_0), 
  ~~~~\Phi(\tv)\wedge \tau(\tv,\tv')\models \tau_0(\tv_0, \tv'_0)
 \end{equation}
\end{definition}

The following fact is immediate:

\begin{proposition}
 Let $\cS_0$ be a sub-system specification of $\cS$ like in Definition~\ref{def:sub}; then the map $\pi_{\cS_0}$ associating $(\tv_0)_{\cM_{\vert \Sigma_0}, \cI_{\vert \tv_0}}$ to $\val{\tv}$ is a simulation of $T^{\cS}$ by $T^{\cS_0}$ (called a \emph{projection simulation} over  $\Sigma_0, \tv_0$). 
\end{proposition}

Projection simulations are ordered according to the ordering of the simulations of $\cS$
they produce, i.e. we say that $\cS_0$ is \emph{stronger} or \emph{better} than $\cS'_0$ iff $\cT^{\cS_0}\leq \cT^{\cS'_0}$. Once $\Sigma_0, \tv_0$ are fixed, one may wonder whether 
there exists the best projection simulation over $\Sigma_0, \tv_0$. The following straightforward result supplies a (practically useful) sufficient condition:

\begin{proposition}\label{prop:simulation} Let $
 \cS=(\Sigma, \tv, \Phi, \iota, \tau, AP)
 $ be a system specification, let $\Sigma_0$ be a subsignature of $\Sigma$ and let $\tv_0\subseteq \tv$ be $\Sigma_0$-variables.
 Suppose that there exist $\Sigma_0$-\formulae $\Phi_0(\tv_0), \iota_0(\tv_0), \tau_0(\tv_0, \tv'_0)$ such that (let $\tv:=\tv_0, \tv_1$):
 \begin{description}
  \item[{\rm (i)~\;}] $\models_{\Sigma}\Phi_0(\tv_0) \leftrightarrow \exists \tv_1 \Phi(\tv_0, \tv_1)$;
  \item[{\rm (ii)~}] $\models_{\Sigma} \iota_0(\tv_0) \leftrightarrow \exists \tv_1 \iota(\tv_0, \tv_1)$;
  \item[{\rm (iii)}] $\models_{\Sigma} \tau_0(\tv_0, \tv'_0)
  \leftrightarrow \exists \tv_1 \,\exists \tv'_1 (\Phi(\tv_0, \tv_1)\wedge \tau(\tv_0, \tv_1, \tv_0', \tv_1'))$.
 \end{description}
 If we let $\cS_0$ be the subsystem specification $(\Sigma_0, \tv_0, \Phi_0, \iota_0, \tau_0, AP)$,
  then  the projection simulation $\pi_{\cS_0}$ is the best projection simulation over $\Sigma_0, \tv_0$.
 \end{proposition}
 \begin{proof}
 That $\cS_0= (\Sigma_0, \tv_0, \Phi_0, \iota_0, \tau_0, AP)$ is a subsystem specification of $\cS$ is clear; let us now pick  another 
 subsystem specification $\cS'= (\Sigma_0, \tv_0, \Phi', \iota', \tau', AP)$ of $\cS$ inducing a projection simulation over the same subsignature
 $\Sigma_0$ and the same sub-tuple of variables $\tv_0$. According to~\eqref{eq:sub}, we have
 \begin{equation*}
  \Phi(\tv)\models_{\Sigma} \Phi'(\tv_0), ~~~~\iota(\tv)\models \iota'(\tv_0), 
  ~~~~\Phi(\tv)\wedge \tau(\tv,\tv')\models \tau'(\tv_0, \tv'_0)
 \end{equation*}
 that is
 \begin{equation*}
  \Phi_0(\tv_0)\models_{\Sigma} \Phi'(\tv_0), ~~~~\iota_0(\tv_0)\models \iota'(\tv_0), 
  ~~~~ \tau_0(\tv,\tv')\models \tau'(\tv_0, \tv'_0)
 \end{equation*}
 which guarantees that $\cT^{\cS_0}\leq \cT^{\cS'_0}$.
\end{proof}

 To understand the meaning of the above proposition, one should keep in mind that there is no reason why the $\Sigma$-\formulae $\exists \tv_1 \Phi, \exists \tv_1 \iota$ and 
 $\exists \tv_1 \,\exists \tv'_1 (\Phi\wedge \tau)$
should be equivalent to $\Sigma_0$-\formulae  (in our applications, $\Sigma_0$ contains only
the sort and the symbols of linear first-order arithmetic,
so no higher order variables are allowed in $\Sigma_0$-\formulae). Thus, the road map to apply Proposition~\ref{prop:simulation}
 is to prove some \emph{quantifier-elimination} results in order to find $\Sigma_0$-\formulae equivalent to $\exists \tv_1 \Phi, \exists \tv_1 \iota, \exists \tv_1 \exists \tv'_1 (\Phi\wedge \tau)$.

\subsection{Counter Abstractions for Parameterized Systems}\label{sec:counters}

We now give a closer look at the signatures we need for modeling parameterized systems (i.e. systems composed by 
a finite - but arbitrary! - number of indistinguishable processes). 
We fix a constrained signature $\Sigma$ for the remaining part of the paper. Such $\Sigma$ should be adequate for
modeling parameterized systems, hence we assume that $\Sigma$ consists of:
\begin{itemize}
 \item[{\rm (i)}] the integer sort $\mathbb{Z}$, together with some parameters (i.e. free individual constants) as well as all operations and predicates of  linear arithmetic 
  (namely, $0, 1, +, -, =, <, \equiv_n$);
  \item[{\rm (ii)}] the enumerated truth value sort $\Omega$, with the constants $\mathtt{tt,ff}$ and the Boolean operations on them;
  \item[{\rm (iii)}] a finite sort $\Proc$, whose cardinality is constrained to be equal to the arithmetic
  parameter $\mathtt{N}$ (this sort models the processes - all identical to each other - taking part 
  in our parameterized system as actors); equality is the only predicate/function symbol defined on this sort;
  \item[{\rm (iv)}] further enumerated sorts $\Data$, modeling local status, local flags, etc.
\end{itemize}
 
The subsignature $\Sigma_0$ comprising only the items (i)-(ii) above is called the \emph{arithmetic subsignature} of $\Sigma$; 
the subsignature $\Sigma_2$ comprising only the items (ii) and (iv) above is called the \emph{data subsignature} of $\Sigma$.
Below, besides \emph{integer} variables (namely variables of sort $\mathbb{Z}$), \emph{data} variables (namely variables of sort $\Data$) and \emph{index}  variables (namely variables of sort $\Proc$), we use two other kinds of variables, that we call \emph{enumerated} and \emph{arithmetic   array-ids}: an enumerated array-id is a variable of type $\mathtt{\Proc\to \Data}$ and an arithmetic array-id is a variable of
 type $\mathtt{\Proc\to \mathbb{Z}}$.

Let  now $\cS=(\Sigma, \tv, \Phi, \iota, \tau, AP)$ be a system specification based on the above signature $\Sigma$.
The variables $\tv$ of $\cS$  include some integer variables $\tv_0$ and in addition variables for arithmetic and enumerated arrays-ids.
Let us suppose that $\tv=\tv_0\tv_1$, where 
$\tv_1$ is the tuple of array variables
and the $\tv_0$ are all the  integer variables of the system. We suppose also that the \formulae in $AP$ - namely the \formulae expressing observable properties - are all open $\tv_0$-\formulae (in particular, they are all
$\Sigma_0$-\formulae, where $\Sigma_0$ is the arithmetic subsignature of $\Sigma$).

Let $\cS=(\Sigma, \tv_0\tv_1, \Phi, \iota, \tau, AP)$ be as above. A \emph{counter abstraction} of $\cS$ is a subsystem specification of the kind $\cS_0=(\Sigma_0, \tv_0, \Phi_0, \iota_0, \tau_0, AP)$;
counter abstractions are ordered according to the ordering of the simulations of $\cS$
they produce, i.e. we say that $\cS_0$ is stronger than $\cS'_0$ iff $\cT^{\cS_0}\leq \cT^{\cS'_0}$. We are interested in sufficient conditions on $\Phi, \iota, \tau$
ensuring the existence of a strongest counter abstraction. We describe below the sufficient conditions for which we have a first implementation (for stronger conditions, requiring heavier machinery, see~\cite{prep}).

Below we use notations like $\phi(\ux), t(\ux), \dots$ to mean that the formula $\phi$, the term $t, \dots$ contains at most the free variables in the tuple $\ux$; notice also that, since there are no operation symbols defined on the sort $\Proc$, all $\Proc$-atoms~\footnote{
By a $\Proc$-atom (resp. $\Data$-atom)  we mean an atomic formula whose root predicate is applied to terms denoting an element of sort $\Proc$ (resp. $\Data$).
} must be equalities between $\Proc$-variables; for the same reasons, all subterms involving arrays-ids are flat, i.e.  must be of the kind $a(i)$ where $i$ is a variable of sort $\Proc$. Since $\Data$ is enumerated, all  $\Data$-atoms  must be of the kind
$a(i)=b(k)$ or $a(k)=\mathtt{a_i}$, where $a,b$ are enumerated arrays-ids, $i,k$ are $\Proc$-variables, and $ \mathtt{a_i}$ is a constant for a value of type $\Data$.\footnote{Atoms of the kind $\mathtt{a_i}=\mathtt{a_j}$ are equivalent to $\mathtt{ff}$ or to $\mathtt{tt}$ because enumerated values are assumed to be distinct. }
We call \emph{$\Data$-formula} a Boolean combination of $\Data$-atoms; we also call \emph{extended arithmetic term} a term of type $\mathbb{Z}$ which is an arithmetic parameter, a numeral, an arithmetic variable, a term of the kind $a(i)$ (where $a$ is an arithmetic array-id and $i$ a $\Proc$-variable) or a term 
of the kind $\sharp\{ k\mid \psi(k)\}$, where $\psi(k)$ is a $\Data$-formula in which only the single $\Proc$-variable $k$ occurs. An \emph{extended arithmetic atom} is a formula obtained from extended arithmetic terms by applying to them the arithmetic operations $+,-$ and the arithmetic predicates $=, <, \leq,\equiv_n$.

\begin{theorem}\label{thm:main}
The system specification $\cS=(\Sigma, \tv_0\tv_1, \Phi, \iota, \tau, AP)$ has a strongest (computable) counter abstraction
in case
 $\Phi, \iota, \tau$  are 
disjunctions of \formulae of the kind 
\begin{equation}\label{eq:ss}
\forall i\; \phi(i)
\end{equation}
where $\phi(i)$  is a Boolean combination of $\Data$-atoms and of extended arithmetic atoms (both containing just the  $\Proc$-variable $i$).
\end{theorem}

\begin{proof} In view of Proposition~\ref{prop:simulation}, it is sufficient to show that if $a_1, \dots, a_n$ are array-ids
and 
$\forall i\; \phi(i)$ is a formula like~\eqref{eq:ss}, then $\exists a_1\cdots \exists a_n \forall i \phi(i)$ is equivalent to a formula in pure Presburger arithmetic.

We first show how to eliminate an existential
arithmetic array-id quantifier $\exists a$. This is eliminated (in favour of an extra existentially quantified arithmetic variable) 
by reverse skolemization~\cite{JAR15,reynolds15}: one observes that $\exists a\, \forall i\, \phi(i)$ is equivalent to the formula $\forall i\, \exists x\, \phi(i, x/a(i))$ (see the above observation about the `flatness' of array-ids terms). Then the extra arithmetic existentially quantified variables introduced above are eliminated via Presburger quantifier elimination (notice that 
they do not occur inside $\Data$-atoms or inside abstraction like terms $\sharp\{ k\mid \psi(k)\}$, because $\psi$ is a $\Data$-formula). 

Finally, enumerated array-ids quantifiers can be eliminated using the methods of~\cite{fmsd17}. Alternatively, since only arithmetic variables and enumerated array-ids are left at this point, it is also possible to make a BAPA-encoding and to use  the  quantifier elimination procedure  for BAPA~\cite{KunkakJAR}. Such encoding can be obtained as follows. Notice that $\Data$-atoms involving only the 
variable $j$ can be written as $a(j)= \mathtt{a_k}$ for some enumerated value $\mathtt{a_k}$;\footnote{ 
Atoms like $a(j)= b(j)$ can be eliminated via $\bigvee_k (a(j)= \mathtt{a_k}\wedge b(j)= \mathtt{a_k})$.
} thus if we introduce set variables $S_{a, \mathtt{a_k}}$ for the sets $\{j \mid a(j)= \mathtt{a_k}\}$ , we can write the terms $\{ j\mid \psi(j)\}$ as Boolean combinations of these set variables
$S_{a, \mathtt{a_k}}$. Finally, if $\psi(i)$  is a Boolean combination of $\Data$-atoms and of extended   arithmetic atoms
without arithmetic array-ids, extended arithmetic atoms can be abstracted out of 
 $\forall i\, \psi(i)$  by `guessing' which of them hold (formally, we introduce a big disjunction, indexed by all Boolean assignments to such extended arithmetic atoms) and, when $\psi(i)$ is reduced to a $\Data$-formula, $\forall i\, \psi(i)$ is equivalent to $\mathtt{N}= \sharp\beta$, where $\beta$ is a Boolean combination of the $S_{a, \mathtt{a_k}}$ introduced above.
\end{proof}

\section{An Example}\label{app:onethird}

In this section, we show how to apply the four-step methodology presented in Subsection~\ref{subsec:strategy} to a concrete problem. All results below have been certified via our prototype
\textsc{ARCA_SIM} explained in Section~\ref{sec:implementation} below.

The One-Third (OT) algorithm is designed to reach agreement in presence of benign transient faults~\cite{heardof,Biely07,sharpie};
 the specification is  reported in Algorithm~\ref{OTcode}.  The protocol is supposed to work with an unlimited number of failures, but failures are supposed to be transient (processes may behave correctly in some rounds and not correctly in other rounds) and benign (processes, if they send any value, they send their own real value - which might or might not be received by the others - i.e. no fake value is sent or received). 
 To be able to apply our techniques, \emph{we need the extra assumption that the value to be agreed on is 
 taken from a finite preassigned set} - let it be $\{0,1\}$ for simplicity. 
We apply our  four-steps plan.

\begin{algorithm}[t]
\begin{footnotesize}
\begin{tabbing}
aaaa\=bb\=cccc\=dddd\=eeee\=ffff   \kill
{\bf Round $k$:} each process executes the following \\
\> \>      \textbf send $val$ to all;\\
\> \> \textbf{if} received values from more than $2\mathtt{N}/3$ distinct
processes  \\
\> \> \>      \textbf{then} set $val$ to the smallest most often received value; \\
\> \> \textbf{if}  more than $2\mathtt{N}/3$ received values equal to $val$,  \\
\> \> \>     \textbf{then} accept $val$.
\end{tabbing}
\end{footnotesize}
\caption{\label{OTcode}One-Third Algorithm:}
\end{algorithm}

\textbf{Step (1):} \emph{we  produce a formalization in higher order logic}. We employ:
\begin{description}
 \item[-] an array-id $V: \mathtt{Proc} \longrightarrow \{0,1\}$ ($V(x)$ is the value currently held by $x$);
 \item[-] an array-id $A: \mathtt{Proc} \longrightarrow \{\bot,0,1\}$ ($A(x)$ is the value accepted by $x$, initially $A(x)=\bot$);
 \item[-] arithmetic array-ids $R_0, R_1$ ($R_0(x)$ is the number of 0-values received by $x$ and $R_1(x)$ is the number of 1-values received by $x$).
\end{description}

We initialize the system using the following formula $\iota$:
\begin{equation}\label{eq:initOT}
 \mathtt{N}>2~\wedge~\forall x~A(x)=\bot
\end{equation}
(the assumption $\mathtt{N}>2$ is not needed, but produces a more readable output).
The transition relation is specified by the formula $\tau$ below:
$$
\begin{aligned}
& ~~~~~\forall i~[0\leq R_0'(i) \leq \sharp \{ x \mid V(x)=0\} \wedge 0\leq R_1'(i) \leq \sharp \{ x \mid V(x)=1\}] ~~~~~~\wedge   
\\ &
\wedge ~\forall i \left[ 
 \begin{aligned}
  & ~~~~~(R_0'(i) + R'_1(i) > 2 \mathtt{N}/3 \wedge R_1'(i)> R_0'(i) \wedge V'(i)=1) ~\vee 
  \\
  &~\vee~ (R_0'(i) + R'_1(i) > 2 \mathtt{N}/3 \wedge R_0'(i)\geq R_1'(i) \wedge V'(i)=0)  ~\vee 
  \\
  & ~\vee~ (R_0'(i) + R'_1(i) \leq 2 \mathtt{N}/3 \wedge V'(i)= V(i))  
 \end{aligned} 
 \right]~~~~~~\;\wedge 
 \\ &
 \wedge ~\forall i \left[ 
 \begin{aligned}
  & ~~~~~(R'_1(i) > 2 \mathtt{N}/3 \wedge A'(i)=1) ~\vee 
  \\
  &~\vee~ (R_0'(i)  > 2 \mathtt{N}/3 \wedge  A'(i)=0)  ~\vee 
  \\
  & ~\vee~ (R_0'(i) \leq 2 \mathtt{N}/3 ~\wedge~ R'_1(i) \leq 2 \mathtt{N}/3 \wedge A'(i)= A(i))  
 \end{aligned} 
 \right]
\end{aligned}
$$
As usual, the primed variables $R'_0, R'_1, V', A'$ denote the updated values of the arrays $R_0, R_1, V, A$ (the arrays $R_0, R_1$ actually do not occur in $\tau$, because the 
update of the local status of the processes only depends on the messages received in the current round - and the numbers of such messages are stored in $R'_0, R'_1$).
Notice that the formula $\tau$ matches the syntactic requirements of Theorem~\ref{thm:main} (just swap the universal quantifier $\forall i$ and the conjunctions).

\vskip 2mm
\textbf{Step (2):} \emph{we manually add counters to our specification}. We introduce six counters, namely 
\begin{equation*}
\begin{aligned}
 & z_{00}=\sharp\{i \mid A(i)=0 \wedge V(i)=0\}, 
 & z_{10}=\sharp\{i \mid A(i)=1 \wedge V(i)=0\},~~ 
 \\
 & z_{\bot 0}=\sharp\{i \mid A(i)=\bot \wedge V(i)=0\},  
 & z_{01}=\sharp\{i \mid A(i)=0 \wedge V(i)=1\}, ~~ 
 \\
 & z_{11}=\sharp\{i \mid A(i)=1 \wedge V(i)=1\},
 & z_{\bot 1}=\sharp\{i \mid A(i)=\bot \wedge V(i)=1\}  ~
\end{aligned}
\end{equation*}
(notice that the counters $z_{10}, z_{01}$ will in fact be constantly equal to 0 during a system run, but we do not assume that this is known in advance).
The \formulae $\iota, \tau$ are modified as follows
\begin{equation}\label{eq:plusOT}
 \iota^+:\equiv~\iota~\wedge~ \delta,
\qquad
 \tau^+:\equiv~\tau~\wedge~ \delta~\wedge \delta', 
\end{equation}
where the auxiliary \formulae $\delta, \delta'$ are the counters definitions supplied in Figure~\ref{fig:macrosOT}.
Our system specification  
\begin{eqnarray*}
&
{\mathcal{S}} =(\Sigma, \{V, A, R, z_{00}, z_{10}, z_{\bot 0}, z_{01}, z_{11}, z_{\bot 1} \}, \Phi, \iota^+, \tau^+)~~
\end{eqnarray*}
 is now complete (we do not need any invariant, so we take $\Phi$ to be $\top$).

\vskip 2mm
\textbf{Step (3):} using the quantifier elimination procedure of Theorem~\ref{thm:main}, \emph{we get rid of higher order variables and  we  compute the projected system}  
\begin{eqnarray*}
&
{\mathcal{S}_0} =(\Sigma_0, \{z_{00}, z_{10}, z_{\bot 0}, z_{01}, z_{11}, z_{\bot 1} \}, \Phi_0, \iota_0, \tau_0)~~.
\end{eqnarray*}
We give the result produced by \textsc{ARCA\_SIM}, with some hand-made manipulations  aiming at making the output more human-readable (all such manipulations are up to logical equivalence).
We use  the auxiliary \formulae from Figure~\ref{fig:macrosOT}.
The formula $\Phi_0$ turns out to be $\top$, whereas $\iota_0$ is 
$$
\psi~ \wedge~ z_{00}= 0~ \wedge~ z_{01}= 0 ~\wedge~ z_{10}= 0 ~\wedge~ z_{11}= 0 ~\wedge~ \mathtt{N}= z_{\bot 0}+z_{\bot 1} ~ \wedge~ 
\mathtt{N}> 2~.
$$
In order to introduce $\tau_0$, we need some extra notation. If $u$ is an arithmetic term involving our counters, we let $u'$ be the same term in which all 
counter variables are primed; we let also $\Delta(u)$ be $u'-u$ and $Decr(u)$ (resp. $Incr(u)$) be $\Delta(u)\leq 0$ (resp. $\Delta(u)\geq 0$). Thus, for instance, $\Delta(z_{00})$ is $z'_{00}-z_{00}$
and $Incr(z_{00})$ is $z'_{00}-z_{00}\geq 0$.
Now $\tau_0$ is the conjunction of
$\psi\wedge \psi'$ (where $\psi$ and $\psi'$ are as defined in Figure~\ref{fig:macrosOT}) with the disjunction of the 7 \formulae below:
$$
\begin{aligned}
& \neg p_0~\wedge \neg p_1~\wedge\neg p_2~\wedge p_3~\wedge~ Incr(z_{00})\wedge~ Incr(z_{\bot 0})\wedge~ Incr(z_{10})~\wedge~
\\ &
\wedge~\Delta(z_{00}+z_{01})= 0 ~\wedge~ \Delta(z_{10}+z_{11})= 0~\wedge~
\Delta(z_{\bot 0}+z_{\bot 1})= 0
\\
&\\
& \neg p_0~\wedge \neg p_1~\wedge p_2~\wedge \neg p_3~\wedge Decr(z_{00})\wedge~ Decr(z_{\bot 0})\wedge~ Decr(z_{10})~\wedge~
\\ 
& \wedge~\Delta(z_{00}+z_{01})= 0 ~\wedge~ \Delta(z_{10}+z_{11})= 0~\wedge~
\Delta(z_{\bot 0}+z_{\bot 1})= 0
\\
&\\
& \neg p_0~\wedge \neg p_1~\wedge p_2~\wedge p_3~\wedge \Delta(z_{00}+z_{01})= 0 ~\wedge~ \Delta(z_{10}+z_{11})= 0~\wedge~
\\
&\wedge~
\Delta(z_{\bot 0}+z_{\bot 1})= 0
\\
& \\
& \neg p_0~\wedge  p_1~\wedge \neg p_2~\wedge \neg p_3~\wedge Incr(z_{00})\wedge~ Decr(z_{01})\wedge~ Decr(z_{10})~\wedge~ 
\\
& \wedge~Decr(z_{11})\wedge~Decr(z_{\bot 0})~\wedge~ Decr(z_{\bot 1})
\\
& \\
& \neg p_0~\wedge  p_1~\wedge \neg p_2~\wedge  p_3~\wedge 
~Decr(z_{\bot 1})
~\wedge~
Decr(z_{01})
~\wedge~
Decr(z_{11})
~\wedge
\\
 &
~\wedge~
Decr(z_{10}+z_{11})
~\wedge~
Incr(z_{00}+z_{01}+z_{10}+z_{11})
\\
& \\
& p_0~\wedge\neg  p_1~\wedge \neg p_2~\wedge \neg p_3~\wedge~
~Decr(z_{00})
~\wedge~
Decr(z_{01})
~\wedge~
Decr(z_{10})
~\wedge~
\\ &
\wedge~
Decr(z_{\bot 0})
~\wedge~       
Decr(z_{\bot 1})
\\
& \\
& p_0~\wedge \neg p_1~\wedge  p_2~\wedge \neg p_3~\wedge
Decr(z_{00})
~\wedge~ 
Decr(z_{10})
~\wedge~ 
Decr(z_{\bot 0})
~\wedge~ 
\\ &
\wedge~
Decr(z_{00}+z_{01})
~\wedge~ 
Decr(z_{\bot 0}+z_{\bot 1})
\end{aligned}
$$

\begin{figure}[t]
\begin{tabular}{lll}
&$\delta:~\equiv~~~~$
&$
z_{00}=\sharp\{i \mid A(i)=0 \wedge V(i)=0\}
\wedge
z_{10}=\sharp\{i \mid A(i)=1 \wedge V(i)=0\}
\wedge
$\\
&&$
\wedge
z_{\bot 0}=\sharp\{i \mid A(i)=\bot \wedge V(i)=0\} \wedge  
 z_{01}=\sharp\{i \mid A(i)=0 \wedge V(i)=1\}
$
\\
&&$
\wedge z_{11}=\sharp\{i \mid A(i)=1 \wedge V(i)=1\}
 \wedge  
 z_{\bot 1}=\sharp\{i \mid A(i)=\bot \wedge V(i)=1\}  
$
\\
&& \\
&$\delta':\;\equiv~~~~$
&$
z'_{00}=\sharp\{i \mid A'(i)=0 \wedge V'(i)=0\}
\wedge
z'_{10}=\sharp\{i \mid A'(i)=1 \wedge V'(i)=0\}
\wedge
$\\
&&$
\wedge
z'_{\bot 0}=\sharp\{i \mid A'(i)=\bot \wedge V'(i)=0\} \wedge  
 z'_{01}=\sharp\{i \mid A'(i)=0 \wedge V'(i)=1\}
$
\\
&&$
\wedge z'_{11}=\sharp\{i \mid A'(i)=1 \wedge V'(i)=1\}
 \wedge  
 z'_{\bot 1}=\sharp\{i \mid A'(i)=\bot \wedge V'(i)=1\}  
$
\\
&&\\
&$\psi:\;\equiv~~~~\;$  
&
$
   ~0\leq z_{00}\leq \mathtt{N} ~\wedge~0\leq z_{10}\leq \mathtt{N} ~\wedge 0\leq z_{\bot 0}\leq \mathtt{N} ~\wedge
$
\\
&&
$
   ~0\leq z_{01}\leq \mathtt{N} ~\wedge~0\leq z_{11}\leq \mathtt{N} ~\wedge 0\leq z_{\bot 1}\leq \mathtt{N} ~\wedge
$
\\
&& $
   ~\mathtt{N}~=~ z_{00}+ z_{10}+ z_{\bot 0} + z_{01}+ z_{11}+ z_{\bot 1}
$
\\
&   &
\\
&$\psi':\;\equiv~~~~\;$  
&
$
   ~0\leq z'_{00}\leq \mathtt{N} ~\wedge~0\leq z'_{10}\leq \mathtt{N} ~\wedge 0\leq z'_{\bot 0}\leq \mathtt{N} ~\wedge
$
\\
&& $
   ~0\leq z'_{01} \leq\mathtt{N} ~\wedge~0\leq z'_{11} \leq\mathtt{N} ~\wedge 0\leq z'_{\bot 1}\leq \mathtt{N}~\wedge
$
\\
&& $
   ~\mathtt{N}~=~ z'_{00}+ z'_{10}+ z'_{\bot 0} + z'_{01}+ z'_{11}+ z'_{\bot 1}
$
\\
&   &
\\
&$p_0:~ \equiv~~~~$ &$t < z_1$ 
\\
&$p_1:~ \equiv~~~~$ &$t < z_0$ 
\\
&$ p_2:~\equiv~~~~$ &$t+1<2z_1~\wedge~ 0<z_0 $
\\
& $p_3:~ \equiv~~~~$ &$t<2 z_0 ~\wedge~ 0 < z_1$
\\
\end{tabular}
\caption{\label{fig:macrosOT} Auxiliary \formulae for OT (we let $t:=\llcorner 2\mathtt{N}/3 \lrcorner$, $z_0:=z_{00}+z_{10}+z_{\bot 0}$
and $z_1:=z_{01}+z_{11}+z_{\bot 1}$).}
\end{figure}

\begin{table}[t]
\begin{center}
\begin{footnotesize}
\begin{tabbing}
aaaa\=bb\=cccc\=dddd\=eeee\=ffff   \kill
{\bf Agreement:}  \\
\> whenever two processes have reached a decision, the values they have decided on must be equal.~~~~~~~~~~~~~~~~~~~~~~~~~~~~~~~~~~~~~~~~\,
\\
{\bf Weak Validity:}  \\
\> if all processes propose the same initial value, they must decide on that value. 
\\
{\bf Irrevocability:}  \\
\> if a process has decided on a value, it does not revoke the decision later.
\end{tabbing}
\caption{\label{fig:propOT}Properties to be certified for OT.}
\end{footnotesize}
\end{center}
\end{table}

\vskip 2mm
\textbf{Step (4):} \emph{we express the safety properties we are interested in using our projected counters and we use an SMT-based tool to check  them}. 
The relevant properties are agreement, 
weak validity and irrevocability (see Table~\ref{fig:propOT}).
\emph{Agreement} can be formalized with our counters by saying that the system never reaches a status satisfying $z_{00}+z_{01}>0 \wedge z_{10}+z_{11}>0$.
\emph{Weak validity} can be tested by checking that the system never reaches a status satisfying $z_{10}+z_{11}>0$, once  
initialized to $\iota \wedge z_{\bot 0}=\mathtt{N}$. \emph{Irrevocability} cannot be fully expressed with our counters, but can be approximated by adding a switch $S$
that is turned to $\mathtt{tt}$ as soon as we have $z_{00}>0$ and then checking that 
the system cannot reach a status satisfying $S=\mathtt{tt}\wedge z_{00}=0$. 
All the above problems can be formulated with a different choice of counters (we employed a maximum choice above); in all variants,\footnote{
The expected obvious property that
$z_{01}+z_{10}$ is always equal to 0 can also be checked by our tool combination.
} \textsc{ARCA\_SIM} takes 
 1-2 seconds to produce the HORN SMT\_LIB file for $\mu Z$ and the latter solves the related fixpoint problem in at most half a second, see the experimental data in Section~\ref{sec:implementation} below.

\section{A First Implementation}\label{sec:implementation}

We implemented the procedure of Theorem~\ref{thm:main} in a prototype tool called \textsc{ARCA\_SIM}. Such tool 
accepts system specifications matching the syntactic restrictions of Theorem~\ref{thm:main} and produces as output a file in the HORN SMT\_LIB format, ready to be model-checked e.g. by $\mu Z$~\cite{muZ}, the fixpoint engine of the SMT solver \textsc{z3}. In successful cases, $\mu Z$ produces an invariant (entirely expressed in terms of our counters) which guarantees the safety of the original system.

A specification file for \textsc{ARCA\_SIM} should first contain \emph{declarations} for parameters, integer variables and arithmetic and enumerated array-ids.
Parameters include a symbol $\mathtt{N}$ denoting the (finite but unknown) number of processes acting in the system; moreover,
with each enumerated array-id, a number $m$ is associated, whose meaning is that of telling the tool that the values of such array-id are taken into
the set $\{0, \dots, m-1\}$. Then \emph{counters definitions} are introduced: these must have the form of  equalities $z = \sharp \{ k\mid \psi(k)\}$, where $\psi$ is a data formula. The system \emph{transition} is given as a single variable universally quantified disjunction of cases $\forall x \bigvee_i \tau_i$, where each $\tau_i$ is specified via a formula of the kind $\phi_{i1} \wedge \phi_{i2}$, where: (i) $\phi_{i1}(x)$ is a conjunction of extended arithmetic atoms (in such atoms, terms like $\sharp \{ k\mid \psi(k)\}$ must have been replaced by the corresponding  counters); (ii) $\phi_{i2}(x)$ is 
a $\Data$-formula.  
%
%
The \emph{initial} formula follows the same syntax as the transition formula (but only one case is allowed), whereas the formula expressing the (negation of the) \emph{safety} property must be an arithmetic formula containing only counters, integer variables and parameters. 

\textsc{ARCA\_SIM} produces a file for $\mu Z$ basically 
 following the proof of Theorem~\ref{thm:main}; it uses a BAPA-quantifier elimination algorithm adapted to the shape 
 of the \formulae arising from our benchmarks. More specifically, the tool proceeds as follows:
\begin{description}
 \item[{\rm (i)}] first, it eliminates (from the arithmetic part $\phi_{i1}$ of  each transition case) the arithmetic array-ids by  reverse skolemization and 
 Presburger quantifier elimination;
 \item[{\rm (ii)}] then, the whole transition is rewritten as a disjunction of \formulae of the kind
 \begin{equation}\label{eq:dis}
  \bigwedge_i (z_i=\sharp \{ k \mid \psi_i(k)\}) \wedge \alpha \wedge \forall k\, \theta(k)
 \end{equation}
 where we have, besides the counter definitions $z_i=\sharp \{ k \mid \psi_i(k)\}$, a Boolean assignment $\alpha$ (seen as a conjunction of literals)  to the arithmetic atoms occurring in the problem, and 
 a single-variable universally quantified $\Data$-formula $\forall k\, \theta(k)$;
 \item[{\rm (iii)}] auxiliary counters are now introduced: we have one counter $z_f$ for each function $f$ associating values to enumerated array-ids
 ($z_f$ counts the cardinality of the set $\{k \mid \bigwedge_a a(k)= f_a
 \wedge \bigwedge_a a'(k)= f_{a'}\}$); the previous counters are expressed as linear combinations of these new counters;
 in addition, in each disjunct~\eqref{eq:dis}, the universally quantified formula $\forall k \theta(k)$  is replaced by the equation $\mathtt{N} = \sum \epsilon_f z_f$, where 
 $\epsilon_f$ is 0 or 1 depending on whether the $\Data$-formula defining $z_f$ is consistent or not with $\theta$;
 \item[{\rm (iv)}] in the final steps, all arithmetic atoms involving old and new counters are collected for each  disjunct~\eqref{eq:dis}; the new counters are eliminated by quantifier elimination and the resulting \formulae give the disjuncts of the transition of the projected counter system.
\end{description}

Contrary to what one might expect, the quantifier elimination steps in (i) and (iv) are not so problematic, because of the special
shapes of the arithmetic \formulae arising from the benchmarks we analyzed. In fact, we did not even 
use a full Presburger quantifier elimination module in \textsc{ARCA\_SIM} for the reasons we are going to explain. In our examples, the quantifier elimination problems in (i) 
involve just easy (`difference bounds'-like) constraints and those in (iv) are usually solved by a substitution (in other words, the formula where a variable $z$ needs to be eliminated from, always contains an equality like $z=t$).\footnote{ 
In case a maximum choice of counters is made by the user, one can even formally prove that this is always the case.
} Notice also that, in case a difficult integer quantifier elimination problem arises, shifting to the (better 
behaved from the complexity viewpoint) Fourier-Motzkin real arithmetic
quantifier elimination procedure is a sound strategy: this is because, in the end, the tool needs to produce just a simulation (i.e. an abstraction). Although \textsc{ARCA_SIM} was prepared to make such a shifting to Fourier-Motzkin procedure, it never did it  during our experiments.

The step (ii) basically amounts to an ``all sat''  problem (i.e. to the problem of listing all Boolean assignments satisfying a formula), which is difficult but can be handled efficiently. The real bottleneck seems to be the need of introducing in (iii) a large amount of auxiliary counters: future work should concentrate on improving heuristics  here. 
Notice that, even in the case the user made an (exponentially expensive) maximum choice of counters, the counters we need in (iii) are even more, because the auxiliary counters in (iii) must take into consideration both the actual and the updated enumerated array-ids (by a `maximum choice of counters' we mean the introduction of a counter for each of the sets $\{k \mid \bigwedge_a a(k)= f_a\}$, varying $f$ among the functions associating values to enumerated array-ids).

\vskip 1mm\noindent\textbf{Some Experiments.}~
In this Subsection we report our first experiments; 
the related files, as well as \textsc{ARCA\_SIM} executables are available at the following link:

\centerline{\url{http://users.mat.unimi.it/users/ghilardi/arca/arcasim.zip}~~.}

\noindent
Unfortunately, for various reasons, the specifications for  the tool \textsc{ARCA} we used in~\cite{fmsd17} for invariant checking and bounded model-checking  are insufficient and not compatible with the specifications accepted by \textsc{ARCA\_SIM}. 
We only analyzed three representative benchmarks: 
(i) the One-Third (OT) algorithm from~\cite{heardof}, whose  formalization is  described in Section~\ref{app:onethird} above; (ii) the Byzantine Broadcast Primitive (BBP) algorithm from~\cite{toueg87}, whose  formalization is described in~\cite{fmsd17}, Section 7.3; (iii) the Send Receive Broadcast Primitive (SRBP) algorithm from~\cite{Srikanth87}, whose formalization is described in~\cite{fmsd17}, Section 6.
%
For each of these benchmarks, we checked the relevant properties mentioned in the literature (for OT also the emptyness of the sets counted by 
$z_{01}, z_{10}$, see Section~\ref{app:onethird}).\footnote{ Relay properties are split into two safety properties, as explained in~\cite{fmsd17}.}

In the table below, we report the time employed by \textsc{ARCA\_SIM} to produce the Horn SMT-LIB problem and the time employed by $\mu Z$ to solve the latter problem. Timings are all in seconds.
We used a PC equipped with Intel Core i7  processor and operating system Linux Ubuntu 16.04 (64 bits).
We also tried (and included in the  distribution) some buggy versions - taken from~\cite{fmsd17} - of the above algorithms; we obtained the expected \texttt{unsat} answer
from $\mu Z$ (with 
performances similar to those  in the above table). Such \texttt{unsat} answers just mean that the system is `possibly unsafe': they do not
certify bugs, because our counters simulations are, in fact, just simulations. Sometimes, with a maximum choice of counters, it is possible to prove (only offline with the actual techniques) that we are in presence of a bisimulation of the original system and in this case an \texttt{unsat} answer reveals the real presence of a bug. 

{\small
\begin{center}
\begin{tabular}[center]{|l|c|c|c|c|} \hline
\textbf{Algorithm} & \textbf{Property} & \textsc{ARCA\_SIM} \textbf{Time} & \textbf{$\mu Z$ Time} & \textbf{Total Time}  \\ \hline
SRBP \cite{Srikanth87} & Correctness & 2.68  &  0.09   &  2.77    \\ \hline
SRBP \cite{Srikanth87} & Unforgeability & 2.73  & 0.06  &  2.79   \\ \hline
SRBP \cite{Srikanth87} &  Relay I & 2.68 & 0.06 &   2.74 \\ \hline
SRBP \cite{Srikanth87} & Relay II & 2.72 & 0.03 &  2.32 \\ \hline
BBP \cite{toueg87} & Correctness & 3.20 & 0.03 & 3.23   \\ \hline
BBP \cite{toueg87} & Unforgeability & 3.23 & 0.07 & 3.30   \\ \hline
BBP \cite{toueg87} & Relay I & 3.21 & 0.02 &  3.23 \\ \hline
BBP \cite{toueg87}  & Relay II & 3.21 & 0.13 & 3.34  \\ \hline
OT \cite{heardof} & Agreement & 0.76 & 0.26 & 1.02  \\ \hline
OT \cite{heardof} & Weak Validity  & 0.76 & 0.02 & 0.78 \\ \hline
OT \cite{heardof} &  Irrevocability  & 2.03 & 0.42 & 2.45 \\ \hline
OT \cite{heardof} & Empty Counters & 0.24 & 0.11 &  0.35 \\ \hline
\end{tabular}
\end{center}
}

\section{Conclusions}\label{sec:conclusions}

We introduced a  technique for automatically building  counter simulations: the technique consists in modeling  system specifications in higher order logic, then in introducing counters for definable sets and finally in exploiting quantifier elimination results to get rid of higher order variables.
Such technique is quite flexible and 
since, whenever it applies,  it always supplies the \emph{best} simulation, it should be in principle capable  to cover all results obtainable via counter abstractions. 
We underline some further important specific features of our approach.

First of all, the approach is purely \emph{declarative}: our starting point is the informal description of the algorithms (e.g. in some pseudo-code) and the first step we propose is a direct translation into a standard logical formalism (typically, classical Church type theory), without relying for instance  on ad hoc automata devices
or on ad hoc specification formalisms. We believe that this choice can ensure flexibility and portability of our method.

Secondly, the amount of human interaction we require is nevertheless very limited and \emph{confined to design choices}: although the final outcome of our investigations 
should be the integration of our techniques into some logical framework, the key leading to their success relies almost entirely on results (satisfiability and quantifier elimination algorithms) belonging to the realm of decision procedures. 

A delicate point is related to the \emph{syntactic limitations} we require on the \formulae describing system specifications (see the statement of Theorem~\ref{thm:main}): such syntactic limitations are needed to ensure higher order quantifier elimination. Although it seems that a significant amount of benchmarks are captured despite such limitations, it is essential  to develop techniques applying in more general cases. In fact,
Theorem~\ref{thm:main} can be extended in various directions~\cite{prep}; in particular, extensions covering specifications with \formulae containing an extra layer of existentially quantified variables of sort $\Proc$ cover classical benchmarks like those in~\cite{tacas06} and look to be relatively easily implementable.

The \emph{integration of the methodology explained in this paper with proof assistants} is another interesting challenge to be pursued; such integration could on one hand \emph{double-check the invariants and the related proof certificates given  by the SMT-solvers} and on the other hand \emph{use counters invariants supplied by our techniques as lemmata} inside complex interactive verifications tasks.

To conclude, we mention some recent work on the verification of fault-tolerant distributed systems, starting with our own previous work.
The additional original contributions with respect to our previous paper~\cite{arca} and its journal version~\cite{fmsd17} are due to the fact 
that in this paper we  moved from bounded model-checking and invariant checking to the much more challenging task of full model-checking via \emph{invariant synthesis}. As discussed in~\cite{fmsd17} (Section 7), standard model-checking techniques are difficult to apply 
in the present context of fault-tolerant distributed systems  because Pre- and Post-image computations are very expensive and lead to fragments for which full decision procedures 
seem not to be 
available. This is why we tried a different approach, via counter simulations. 

Papers~\cite{konnov2013A,konnov2013B,KonnovVW15}
represent a very interesting and effective research line (summarized in~\cite{KonnovVW15}), where cardinality constraints are not directly handled but abstracted away 
using  counters. In this sense, this research line looks similar to the methodology we applied in this paper (and in contrast to the
alternative methodology we adopted in our previous paper~\cite{arca}); however  abstraction in~\cite{KonnovVW15} and in related papers is not obtained via logical formalizations and  quantifier elimination, but via
a special specification language (`parametric  Promela') and/or via special devices, called `threshold automata'.
%
A comparison with the counter systems we obtain is not immediate and not always possible because the authors of~\cite{KonnovVW15} work on 
asynchronous (not round-based) versions of the algorithms and because their method  suffers of some lack of expressiveness
 whenever local counters are unavoidable. On the other hand, they are able to certify also liveness properties, whereas at 
the actual stage we can only do that  by making reductions (whenever possible) to safety or bounded model checking problems.

Paper~\cite{sharpie} directly handles cardinality constraints for interpreted sets by employing specifically tailored abstractions and some incomplete inference schemata at the level of the decision procedures. Nontrivial  invariant properties are synthesized and checked, 
based on Horn constraint solving technology; this is the same technology we rely on in our final step, however the counter systems 
we get are `as accurate as possible', in the sense that they supply `the best simulations' as stated in Theorem~\ref{thm:main}. 

Paper~\cite{dragoj} introduces an expressive logic, specifically tailored to handle consensus problems (whence the name `consensus logic' $CL$). Such logic employs arrays with values into power set types, hence it is naturally embedded 
 in a higher order logic context.  Paper~\cite{dragoj} is not concerned with simulations and bisimulations, rather it uses an incomplete algorithm in order to certify   invariants. A smaller fragment (identified via several syntactic restrictions) is introduced in the final part of the paper and a decidability proof for it 
is sketched. 

Finally, we mention the effort  made by the interactive theorem proving community in formalizing and verifying fault-tolerant distributed algorithms (see e.g.~\cite{merz}); such approach is a natural complement to ours.

\bibliographystyle{eptcs}
\bibliography{references}

\appendix

\newpage

\end{document}